\documentclass{article}

\usepackage[preprint,nonatbib]{nips_2018}
\usepackage{cite}

\usepackage[utf8]{inputenc} 
\usepackage[T1]{fontenc}    
\usepackage{hyperref}       
\usepackage{url}            
\usepackage{booktabs}       
\usepackage{amsfonts}       
\usepackage{nicefrac}       
\usepackage{microtype}      
\usepackage[cmex10]{amsmath}

\usepackage{graphicx}
\usepackage{subfigure}
\usepackage{mathrsfs}
\usepackage{amsfonts}

\usepackage{amsthm}
\usepackage{dsfont}
\theoremstyle{remark}
\newtheorem{definition}{Definition}

\newtheorem{theorem}{Theorem}[section]
\newtheorem{lemma}{Lemma}[section]

\newtheorem{property}{Property}
\newtheorem{remark}{Remark}
\newtheorem{assumption}{Assumption}
\usepackage{color}
\usepackage{comment}
\usepackage{bm}
\usepackage{algorithm}
\usepackage{algorithmic}

\newcommand{\eat}[1]{}
\newcommand{\stitle}[1]{\vspace{0.0ex}\noindent{\bf #1}}

\newlength\myindent
\newcommand\bindent{%
  \begingroup
  \setlength{\itemindent}{\myindent}
  \addtolength{\algorithmicindent}{\myindent}
}
\newcommand\eindent{\endgroup}

\def\0{\mathbf{0}}
\def\1{\mathds{1}}
\def\a{\mathbf{a}}

\def\Abf{\mathbf{A}}
\def\b{\mathbf{b}}
\def\Bbf{\mathbf{B}}
\def\c{\mathbf{c}}

\def\d{\mathbf{d}}

\def\e{\mathbf{e}}

\def\f{\mathbf{f}}

\def\g{\mathbf{g}}

\def\Gbf{\mathbf{G}}
\def\Hbf{\mathbf{H}}

\def\I{\mathbf{I}}

\def\Obf{\mathcal{O}}

\def\Pd{\mathcal{P}}

\def\S{\mathcal{S}}

\def\u{\mathbf{u}}
\def\vv{\mathbf{v}}

\def\w{\mathbf{w}}

\def\y{\mathbf{y}}
\def\x{\mathbf{x}}

\def\Xbf{\mathbf{X}}

\def\z{\mathbf{z}}

\def\one{{\bf 1}}

\title{Coded Elastic Computing}

\author{
  Yaoqing Yang, Pulkit Grover, Soummya Kar\\
  Carnegie Mellon University\\
  \texttt{\{yyaoqing, pgrover, soummyak\}@andrew.cmu.edu} \\
   \And
   Matteo Interlandi, Saeed Amizadeh, Markus Weimer \\
   Microsoft \\
   \texttt{\{mainterl, saamizad, mweimer\}@microsoft.com} \thanks{Some preliminary results of the paper have been presented at the Workshop on Systems for ML and Open Source Software at NeurIPS 2018 (without conference proceedings). An updated conference version will appear in ISIT 2019.}\\
}

\begin{document}

\maketitle

\begin{abstract}
Cloud providers have recently introduced new offerings whereby spare computing resources are accessible at discounts compared to on-demand computing. Exploiting such opportunity is challenging inasmuch as such resources are accessed with low-priority and therefore can \emph{elastically} leave (through \emph{preemption}) and join the computation at any time. In this paper, we design a new technique called \emph{coded elastic computing} enabling distributed computations over elastic resources. The proposed technique allows machines to leave the computation without sacrificing the algorithm-level performance, and, at the same time, adaptively reduce the workload at existing machines when new ones join the computation. Leveraging coded redundancy, our approach is able to achieve similar computational cost as the original (uncoded) method when all machines are present; the cost gracefully increases when machines are preempted and reduces when machines join. The performance of the proposed technique is evaluated on matrix-vector multiplication and linear regression tasks. In experimental validations, it can achieve exactly the same numerical result as the noiseless computation, while reducing the computation time by 46\% when compared to non-adaptive coding schemes.
\end{abstract}

\section{Introduction}

New offerings from cloud-service providers allow exploiting under-utilized Virtual Machines (VMs) at a fraction of the original cost~\cite{azure-batch,aws-spot}. For example, Azure Batch provides low-priority machines at about one-fifth of the cost of ordinary virtual machines~\cite{azure-batch-price}.
Similarly, Amazon Spot Instances provide machines at market price with a discount which can reach up to 90\% with respect to regular on-demand prices~\cite{amazon-spot}.
Such offerings, however, have the drawback that machines can be preempted at any time if a high-priority job appears.
This, in turn, will surface as a computation failure at the application level.
While common distributed machine learning frameworks are already built with fault-tolerance~\cite{mllib,tensorflow}, they often assume that failures are transient and rare.
Due to this assumption, machine failures are often recovered by a ``stop-the-world'' scheme whereby the entire system is forced to wait until regular execution on the failure machines are restored from previous state (eventually on new machines).
The above assumptions, however, do not necessarily hold for failures due to machines being preempted because (1) these failures are permanent and local data may not be accessible anymore; (2) several machines can be preempted altogether (up to 90\%~\cite{qoop}); (3) these failures add up to transient failures, therefore leading to more frequent disruptions during computation; and (4) the computational framework may need to acquire additional machines to compensate, meaning that data has to be copied on the new machines which will likely become stragglers for the running computation.
In practice, we observed situations at scale where the stop-the-world scheme results in zero computation progress because, by the time a failure is recovered, a new failure occurs. This results in the necessity to build an elastic run-time framework \cite{narayanamurthy2013towards,reef} and related failure-aware algorithms which can continue the computation and flexibly adapt in the presence of failures.
Another possible technique to deal with preemption is to view the preempted machines as erasure-type faults and ignore them. Although machine learning algorithms are robust to small transient faults, simply ignoring the computational results in these permanently-failed preempted machines may result in algorithmic-level performance loss~\cite{narayanamurthy2013towards}.
The influence of ignoring the computation results for the preemption type of faults is also more severe than usual computation faults because the number of failures can be really large~\cite{qoop}.
Similarly, even if the data are redundant in some applications, ignoring partial results may still lead to reduced confidence levels on the prediction accuracy. It may also not be desirable from a customer's perspective who often requires the full dataset to be present during the entire training process in order to achieve the highest accuracy.

In order to deal with the aforementioned problems, in this paper we propose \emph{coded elastic computing}: a novel distributed learning framework allowing users to train their machine learning models over preemptable machines.
In our coded elastic computing framework machines are allowed to arbitrarily join or leave during a distributed iterative learning task thanks to the introduction of redundancy in the computation.
Although coded elastic computing introduces redundancy using ideas from error-correcting codes \cite{lee2018speeding,GC2,dutta2016short,Salman1,yu2017polynomial,Suhas1,yang2017NIPS,fahim2017optimal, dutta2018optimal}, and can let the computation continue when a preemption failure happens like ordinary error-correcting codes, the way it utilizes the coded data to deal with preemption is fundamentally different from existing works. More specifically, coded elastic computing can flexibly change the workload of each machine at runtime based on the number of available machines by selecting to use only a subset of the encoded data in a cyclic fashion. Apart from providing fault-tolerance when machines are preempted, coded elastic computing is also \emph{positively-elastic} in that it can utilize the properties of the coded data to reduce the workload at existing machines flexibly when new machines join the computation. We will show that the coded elastic computing framework can make the computational cost at each machine scale inversely with the number of machines, which leads to linear scaling of theoretical computational cost. The proposed technique is also useful in other applications besides elastic computation when the number of machines needs to be dynamically adjusted during a learning task, such as when the number of machines is tuned as a hyper-parameter, or when the machines have to be reallocated to achieve fairness \cite{hindman2011mesos} between users or the specific need of some users at runtime.

The proposed coded elastic computing technique is tested in the multi-tenancy cluster at Microsoft as an example of the Apache REEF EGC (Elastic Group Communication) framework. Apache REEF is a library that helps develop distributed high-performance applications on top of cluster resource managers such as YARN \cite{vavilapalli2013apache} and Mesos \cite{hindman2011mesos}.
The Apache REEF project provides a set of abstractions and reusable functional blocks to ease the process of building cloud-scale applications. EGC is a distributed communication framework which extends Apache REEF by providing an API allowing to implement elastic computations by chaining fault-tolerant MPI-like primitives. While fault-tolerant distributed applications are usually conducted using the well-known MapReduce-style computation model based on data shuffling~\cite{dean2008mapreduce,mllib}, recent trends in machine learning show that MPI-based computation~\cite{horovod} provides better performance with respect to the former. Nevertheless, MPI-based applications are in general not fault tolerant, whereas EGC tries to unify the benefits of the two worlds. Based on the EGC framework, we test the proposed technique for a coded implementation of linear regression on a real dataset when machines can leave and join the computation. We show that the current technique can obtain the same convergence behavior as ordinary gradient-descent-based algorithms but can elastically allocate the work load based on the number of available machines without moving data at the existing machines. We also compare with other baselines, such as ignore, replication and an existing algorithm called Elastic Distr-BGD \cite{narayanamurthy2013towards} to show the improvement of the proposed technique in terms of the model generalization error.

In this paper, we first present a coded elastic matrix-vector multiplication algorithm to illustrate the main idea. Then, we present the generalizations of the proposed technique in broader applications, namely matrix-matrix multiplication, linear regression, master-free fully-distributed computing, and the training of deep neural networks (see Section \ref{sec:more_applications}). Finally, we validate the approach with a set of experiments. The proposed technique achieves the exact computation result as the noiseless computation (where the number of machines remain fixed), while adaptively changing the workload at each machine. The observed speedup compared to a non-adaptive coded computing baseline can reach up to 50\%. The contributions of the paper are thus summarized in the following:
\begin{itemize}
    \item We formalize the preemption problem in elastic computing and propose a coded elastic data-partitioning framework to deal with it.
    \item We design a computing technique that can adapt the workload at each machine in the presence of elastic events, without moving data at existing machines.
    \item We test the proposed technique using the Apache REEF EGC framework for linear regression and show advantage on real datasets over multiple baselines.
\end{itemize}

\section{Resource-Elastic Coded Distributed Computing}

In this section, we formally define computation elasticity (see Section~\ref{sec:defitnition}), and present coded elastic techniques. We initially focus on the problem of matrix-vector multiplication for having a better theoretical understanding of coded elastic computing. 
Before presenting the algorithm in Section \ref{sec:algorithm}, we introduce the main idea of the paper in Section \ref{sec:coded_partitioning} and Section \ref{sec:coded_elastic_partitioning}. Then, in Section~\ref{sec:achievable} and Section~\ref{sec:lower_bound}, we analyze the proposed techniques and prove that they are indeed elastic according to our definition.

\subsection{Definition of Computation Elasticity}\label{sec:defitnition}

We characterize \emph{elastic events} whereby existing machines can be preempted, and new machines can be added to the computation. A preemption means that machines are taken away and the local data is lost. To formalize the notion of elastic events that we adopt in this paper, we state the following properties that are characteristic of these events:

\begin{property}\label{prop:elasticity_1}
	Which machine(s) to be preempted is decided by the resource allocator and is not known in advance.
\end{property}

\begin{property}\label{prop:elasticity_2}
	The preemption is permanent. However, new machines may join after some unknown time.
\end{property}

\begin{property}\label{prop:elasticity_3}
	If some machines leave or join, the other machines know immediately about which machines leave or join.
\end{property}

The second and third properties differentiate the elastic events from more commonly considered issues of faults and stragglers because (1) new machines can join the computation, and (2) one may adapt the computation scheme instantly after an elastic event and utilize the newly available resources.

Consider the case when a data matrix $\Xbf$ is stored distributedly. Denote by $P$ the initial number of machines. Denote by $P_\text{max}$ the maximum possible number of machines which equals to the overall size of the universe of machines handled by a cluster scheduler. A configuration point is a tuple $(n,m)$ in which $n$ is the number of machines and $m$ is the memory size. Note that $n$ changes over time, e.g., $n$ can exceed $P$. A computation policy of a task at a given configuration point $(n,m)$ is said to be achievable if the policy completes the task using $n$ machines, each with size $m$. A policy is said to be \emph{optimal} if it obtains the optimal tuple $(e,u)$ simultaneously, where $e$ is the number of machine preemptions that can be tolerated (i.e., the exact result can be computed even if there are $e$ preemptions), and $u$ is the size of data that a machine actually selects to use ($u\le m$). This means the memory of size $m$ can be used to store (coded) data, but during computation, we only access a part of it. Note that the optimal tuple ($e^*$,$u^*$) depends on ($n$,$m$). For a given configuration $(n,m)$, we want to minimize $u$ to reduce the memory access time and maximize $e$, the number of tolerable machine failures. Denote by $\mathcal{A}_{n,m}$ the set of computing policies, and $\mathcal{A}^*_{n,m}\subset \mathcal{A}_{n,m}$ the set of optimal computing policies that obtain the optimal tuple $(e^*,u^*)$. One may think that the optimal tuple $(e^*,u^*)$ may not be unique. For example, there may exist a computing scheme that achieves $(e,u)$ such that $e>e^*$ and $u>u*$. Note that our definition of $(e^*,u^*)$ is not in the sense of Pareto optimality. The $e^*$ and $u^*$ are information-theoretical optimal values that are obtained separately. If there is no scheme that achieves the optimal $e^*$ and $u^*$ at the same time, then in this case, $\mathcal{A}^*_{n,m}$ is an empty set.

\begin{definition}(transition compatibility)
A pair of policies ($a$, $a'$) with $a\in \mathcal{A}_{n,m}$ and $a'\in \mathcal{A}_{n',m'}$ is said to be transition compatible if the policy $a$ can be transitioned to policy
$a'$ without having to move or modify the data in the existing machines in the event of a configuration transition ($n$,$m$)$\to$($n'$,$m'$).
\end{definition}
The key to the above definition is that we discourage inter-machine data movement in order to make the elastic configuration transitions non-disruptive to ongoing computation tasks.
\begin{definition}\label{def:transition_family}(fully transition compatibility)
A family of policies $\mathcal{F} = \{a_{n,m}\}$, $n\in \mathcal{N}$, $m\in \mathcal{M}$, is said to be fully transition compatible if every pair ($a_{n,m}$, $a_{n',m'}$) in $\mathcal{F}$ are transition compatible.
\end{definition}

\begin{definition}(optimal fully transition compatibility)\label{def:transition_family_optimal}
A family of fully transition compatible policies $\mathcal{F} = \{a_{n,m}\}$, $n\in \mathcal{N}$, $m\in \mathcal{M}$ is said to be optimal if all policies are optimal, i.e., each policy $a_{n,m}\in \mathcal{F}$ is in $\mathcal{A}_{n,m}^*$ and obtains the optimal tuple $(e^*,u^*)$ for the number of tolerable machine preemptions and the size of the selected data to use.
\end{definition}

Our goal is thus to find conditions under which an optimal fully transition compatible family exists, and to provide explicit characterizations of the transition compatible families. We now present matrix-vector multiplication techniques that can provide a fully transition compatible family of optimal computation policies with fixed memory cost at each machine, i.e., when $m$ is fixed in different transition compatible policies $a_{n,m}$.

\subsection{Coded Data Partitioning in the Presence of Preempted Machines}\label{sec:coded_partitioning}

Assume that in the worst-case of preemption failures, there are at least $L$ machines that remain\footnote{The parameter $L$, or a lower bound of $L$, is needed for exact computation. However, in many machine learning tasks, one can often optimize with a subset of data due to data redundancy. In that case, knowledge of $L$ is not necessary.}. We will show that the parameter $L$ is also equal to the \emph{recovery threshold}, the meaning of which will be clear in this section (Section \ref{sec:coded_partitioning}). In this paper, we consider repeatedly using the same data but with different input vectors. For matrix multiplications, it means that we compute $\Xbf\w_t$ for $t=1,2,\ldots$ for the same $\Xbf$. This computation primitive is applicable in a variety of scenarios, including training linear models \cite{yang2018coded,haddadpour2018straggler}, PageRank \cite{yang2017NIPS}, model-parallel deep neural networks \cite{dutta2018DNN1,dutta2018DNN2} and many machine learning algorithms at the inference stage.

We partition the data matrix $\Xbf$ into $L$ subsets (or equivalently, submatrices obtained by row-wise partition) $\Xbf_1,\Xbf_2,\ldots,\Xbf_L$ of equal size. If the total number of data points is not divisible by $L$, we can use zero-padding. The generator matrix $\Gbf_{P_\text{max}\times L} = (g_{s,k}), s=1,\ldots P_\text{max}$ is predetermined. We initially generate $P$ \emph{coded} data matrices $\Xbf_s^\text{coded},s=1,2,\ldots,P$, $(P>L)$, in which each matrix is a linear combination of the form:
\begin{equation}\label{eqn:linear_combination}
    \Xbf_s^\text{coded} = \sum_{k=1}^{L} g_{s,k} \Xbf_k
\end{equation}
where each $g_{s,k}$ is a random but predetermined coefficient. The $P$ coded data matrices are distributed to $P$ workers. When the number of machines exceeds $P$, we generate new linear combinations based on $\Gbf_{P_\text{max}\times L}$. The generation can be done before the computation, and the data can be stored in the cloud which has a much larger size than the local fast memory.
\begin{lemma}\label{prop:LoutofP}
Suppose we want to compute the matrix-vector product $\Xbf\w$. Suppose the matrix $\Gbf_{P_\text{max}\times L} = (g_{s,k})$ satisfies the property that any $L\times L$ submatrix of $\Gbf_{P_\text{max}\times L}$ is full-rank. Then, any $L$ out of $n$ coded computation results obtained at the $n$ available machines $\Xbf_s^\text{coded}\w,s=1,2,\ldots,n$ are sufficient to recover the original (uncoded) computation results $\Xbf\w$ regardless of the current number of machines $n$.
\end{lemma}
The rank condition in the lemma can be satisfied by a variety of choices of linear coefficients $g_{s,k}$, e.g., if $g_{s,k}$'s are i.i.d. Gaussian random variables. The recovery of the results is through solving $L$ linear systems of the form $\Xbf_s^\text{coded}\w = \sum_{k=1}^{L} g_{s,k} \Xbf_k\w$ for the $L$ different machines that successfully finish the computation. Lemma \ref{prop:LoutofP} is critical for the failure recovery. It essentially shows that no matter which machines are preempted, as long as the number of remaining machines is not smaller than $L$, the \emph{whole information} of the original data is preserved in the remaining machines, and the computation results can be recovered. This is why we call the parameter $L$ the \emph{recovery threshold}. The parameter $L$ is limited by the storage constraint at each machine. The more redundancy we can add to the data, the lower recovery threshold we need, and hence more failures we can tolerate. In our experiments, we consider the case when $P=P_\text{max}$, i.e., the maximum number of machines is the same as the overall number of machines, and use a redundancy factor of $P/L$ = 2. Therefore, we can at maximum tolerate failures when half of the machines are preempted. Note that this is not necessary because one can generate more data and store the extra data on the cloud so that the number of machines $n$ can exceed the initial number of machines $P$. An often-used coding technique is called \emph{systematic code}, in which the linear coefficients satisfy
\begin{equation}\label{eqn:systematic}
    g_{s,k}=\one_{\{s=k\}},\text{ if }s\le L.
\end{equation}
In this case, the coded data at the first $L$ machines $\Xbf_s^\text{coded},s=1,2,\ldots,L$ are the original data $\Xbf_k ,k=1,2,\ldots,L$. This can provide backward-compatibility to switch between coded computing and other uncoded ordinary computing techniques, and at the same time significantly reduce the cost of encoding the data at the preprocessing stage \eqref{eqn:linear_combination}.

\subsection{Elastic Data Partitioning for Elastic Computation by Using Data In a Cyclic Way}\label{sec:coded_elastic_partitioning}

According to Lemma \ref{prop:LoutofP}, as long as the number of machines that are not preempted is greater or equal to $L$, the remaining data using the coded data partitioning can preserve the whole information of the original data. However, when the number of machines is strictly larger than $L$, it becomes redundant to use all the coded data because the data at $L$ machines already preserve the whole information. One may think that this amount of waste is not significant. However, consider the case when the number of machines gradually increases from $L$ to a large number. In this case, any $L$ machines can provide the correct results, but we are not able to utilize the parallel gain if a fixed coding technique is used. The situation can be partially solved if we use some rateless coded techniques \cite{mallick2018rateless}. See Section~\ref{sec:comparison} for a comparison between the proposed technique and rateless coded computing.

To positively utilize all the remaining machines and achieve the parallel computing capabilities of the extra machines, we select to use data in a cyclic fashion as shown in Figure \ref{fig:elastic_computing}.
\begin{figure}[ht!]
  \centering
  \subfigure[Data encoding without further partitioning]{\includegraphics[scale=0.25]{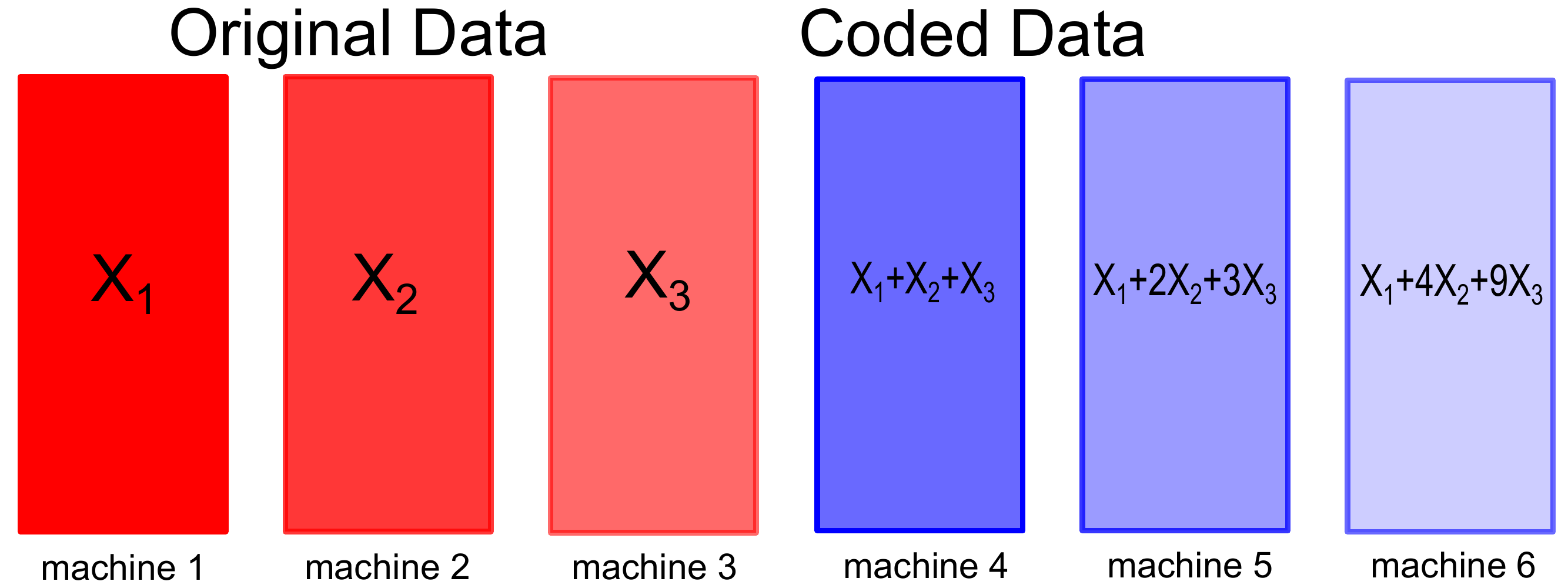}\label{fig:sub1}}\\
  \vspace{-1ex}
  \subfigure[No preempted machine]{\includegraphics[scale=0.25]{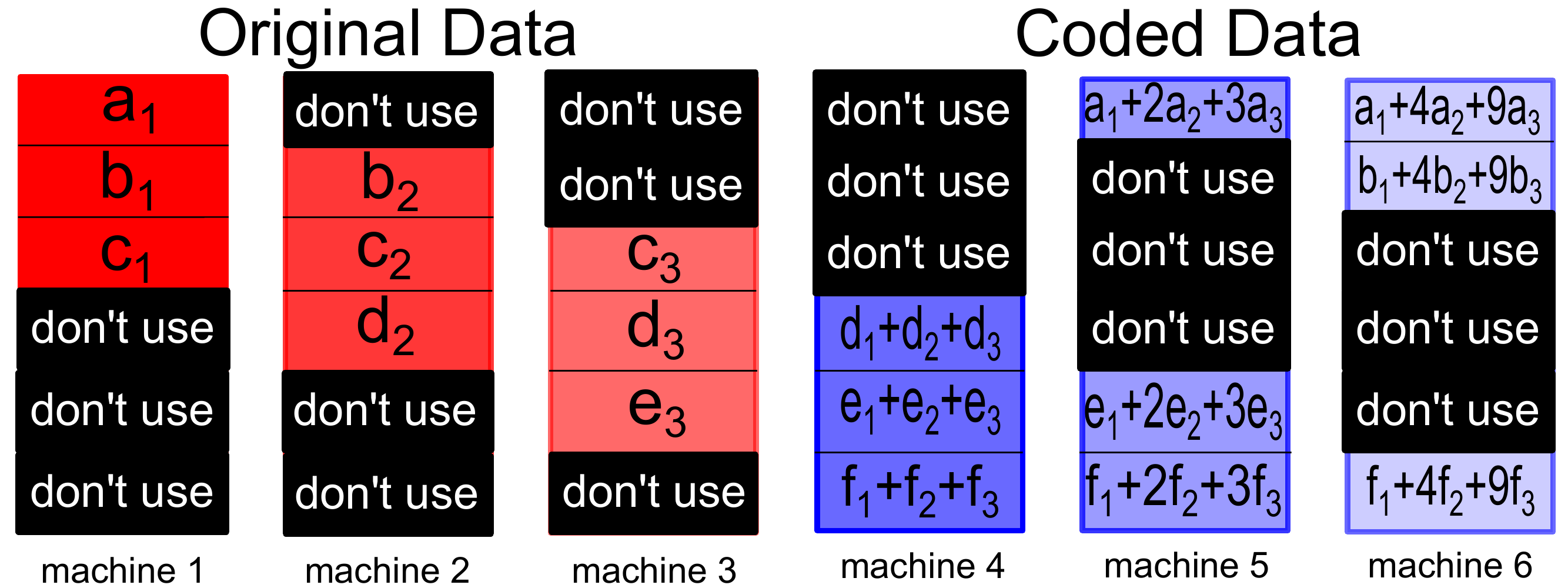}\label{fig:sub2}}\hspace{2ex}
  \subfigure[One preempted machine]{\includegraphics[scale=0.25]{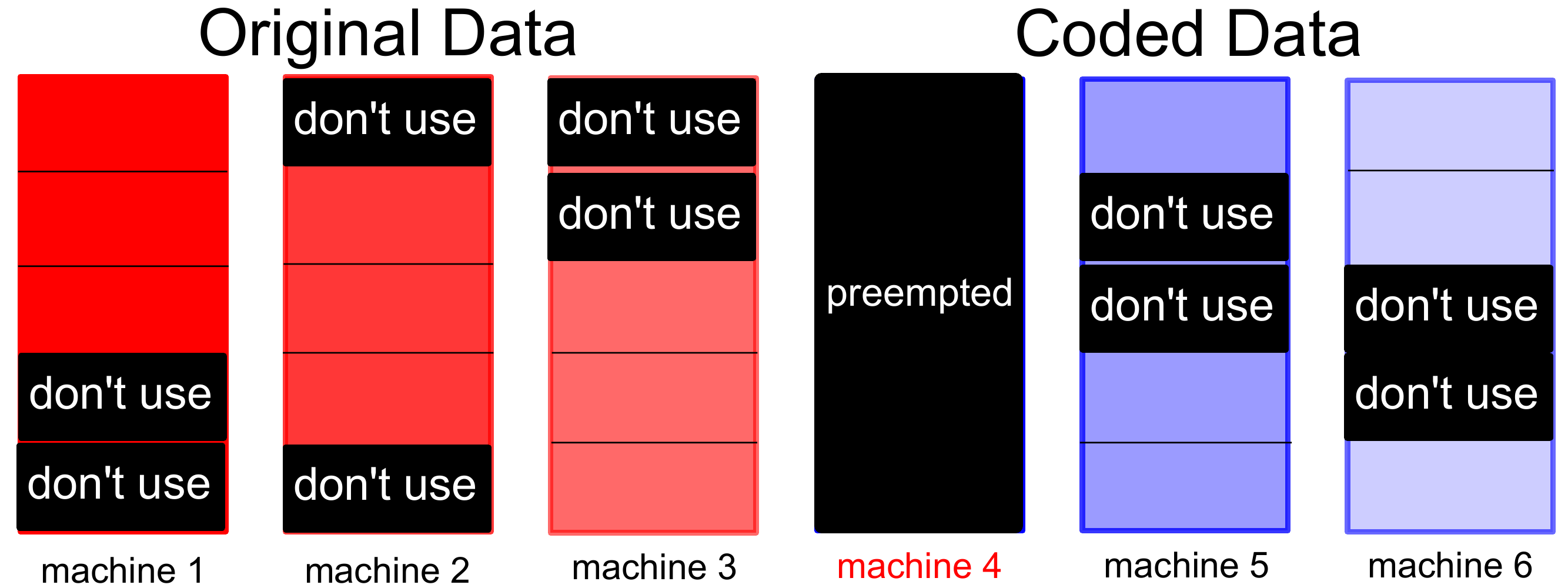}\label{fig:sub3}}\vspace{-1ex}
  \subfigure[Two preempted machines]{\includegraphics[scale=0.25]{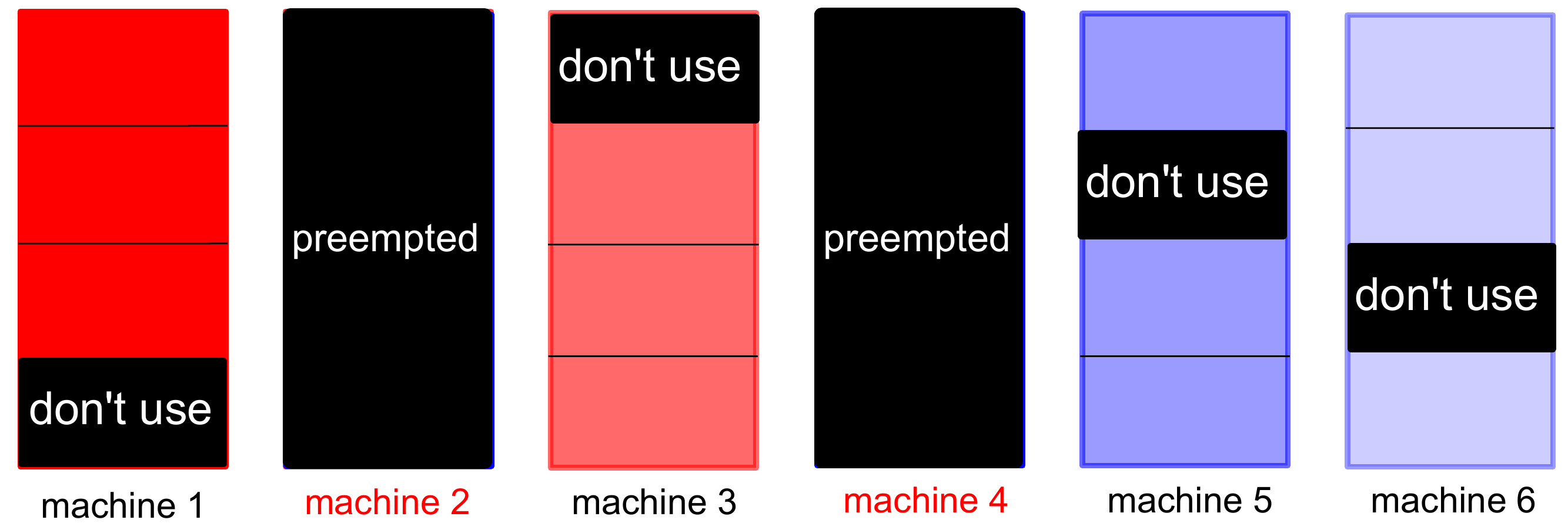}\label{fig:sub4}}\hspace{2ex}
  \subfigure[Three preempted machines]{\includegraphics[scale=0.25]{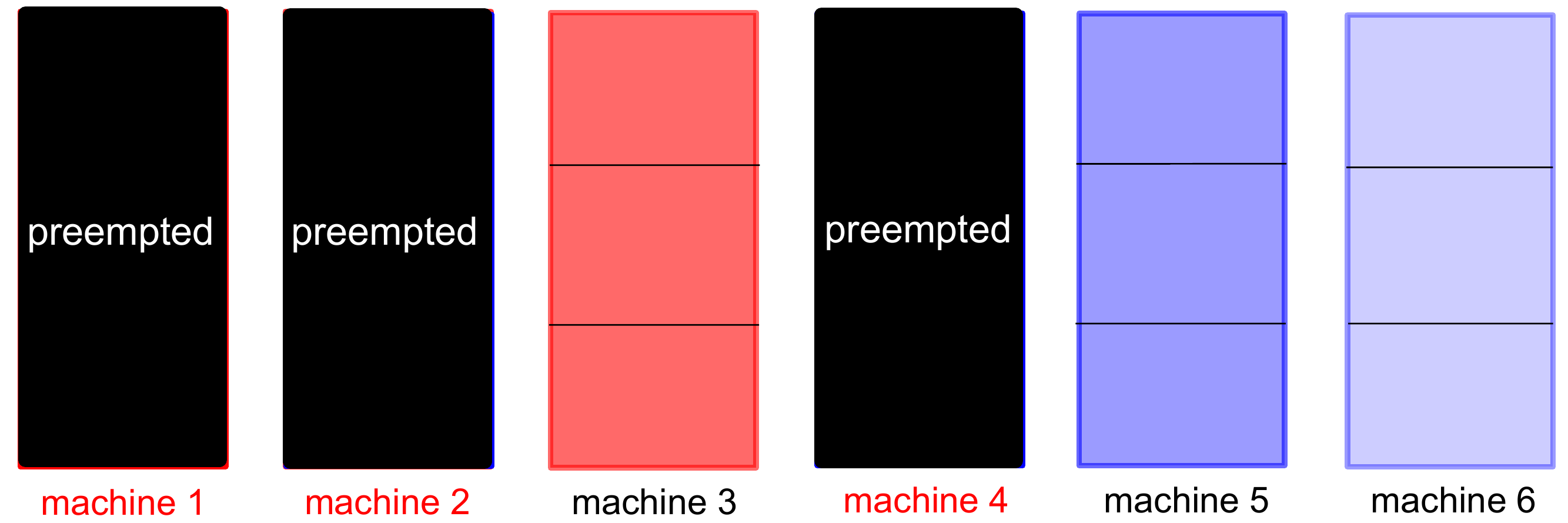}\label{fig:sub5}}\vspace{-2ex}
  \caption{The main idea of elastic data partitioning is to use the data in a cyclic way. The black sub-blocks with ``don't use'' mean that the data is stored but not accessed. Each column of data is stored at one machine. For each group (i.e., row block) of data at different machines, there is enough number of sub-blocks that contain all the information. This cyclic way of using data leads to linear scaling of the per-machine computational cost in the number of machines.\vspace{-2mm}} \label{fig:elastic_computing}
\end{figure}
We use a systematic code (see \eqref{eqn:systematic}) by which the first $L$ of the $P$ coded blocks $\Xbf_s^\text{coded},s=1,2,\ldots,P$ are the original data $\Xbf_k ,k=1,2,\ldots,L$. In Figure~\ref{fig:sub1} we use red to denote original data and blue to denote the remaining coded data. In this example, the initial number of machines is $P = 6$ and the recover threshold is $L=3$. From figure \ref{fig:sub2} to \ref{fig:sub5}, we show how to continue the computation when machines are gradually preempted from 6 to 3 (machines correspond to the columns). The stored data remains fixed i.e., the same way as in Figure~\ref{fig:sub1}, but we further partition the data into smaller blocks and only select part of the data to use. Each machine is initially allocated a single subset of coded data $\Xbf_s^\text{coded},s=1,2,\ldots,P$, among which $L$ subsets are the original data. Each subset of data is represented as a column in any subfigure of Figure \ref{fig:elastic_computing}.

If no failures occur (see Figure \ref{fig:sub2}), to remove redundancy from the data, we partition each data block (column) into $P$ \emph{sub-blocks}, and let each machine only use $L$ out of $P$ sub-blocks. By a sub-block of data, we mean one small rectangle in Figure~\ref{fig:sub2}. When there are $n$ machines and $n\neq P$ (see Figure \ref{fig:sub3}-\ref{fig:sub5}), we partition each data block into $n$ sub-blocks, and still let each node only use $L$ out of $n$ sub-blocks. If new machines join, they download the coded data previously used in some failed machines or some new linearly combined data based on $\Gbf_{P_\text{max}\times L}$. All the machines, including the ones that just join, use elastic data partitioning based on the current number of available machines $n$.

There are two advantages of this type of data usage (1) the overall selected data to use is of the same size as the original data and the selected data across all remaining machines have the same size; and (2) the selected data preserve the whole information of the original data, according to Lemma~\ref{prop:LoutofP}. Thus, we can exactly recover the results while removing the redundancy in the way of using data. These two properties will be formally introduced and proved in Theorem~\ref{thm:achievable}. 

\subsection{Coded Elastic Computing for Matrix-vector Multiplications}\label{sec:algorithm}

We provide the detailed procedures of the coded elastic computing algorithm for the repeated matrix-vector multiplication problem $\Xbf\w_t, t=1,2,\ldots$ in Algorithm \ref{alg:coded_elastic_computing}. We use $\Xbf_{k,j}^{\text{coded}}$ to represent the $j$-th sub-block of the data at the $k$-th machine. We will call $\Xbf_{k,j}^{\text{coded}}$'s with the same $j$ ``the $j$-th group'' of sub-blocks which correspond to the $j$-th row block in any subfigure of Figure~\ref{fig:elastic_computing}. Note that the number of row blocks change with the number of preempted machines. We use $\Gbf_j$ to represent the collection of linear combination coefficients for the $j$-th group (row block) that are selected to use. For example, for the first group (row block) in Figure~\ref{fig:sub2}, we have $
\Gbf_1 = \left[\begin{matrix}
1&0&0\\
1&2&3\\
1&4&9
\end{matrix}\right]$
because the three selected sub-blocks are $\a_1$, $\a_1+2\a_2+3\a_3$ and $\a_1+4\a_2+9\a_3$. And for the last group (row block) in Figure~\ref{fig:sub2}, $\Gbf_6 = \left[\begin{matrix}
 1&1&1\\
 1&2&3\\
 1&4&9
 \end{matrix}\right]$, because the three selected sub-blocks are $\f_1+\f_2+\f_3$, $\f_1+2\f_2+3\f_3$, and $\f_1+4\f_2+9\f_3$.

\begin{algorithm}[!h]
   \caption{Coded Elastic Computing for Matrix-Vector Multiplication}\label{alg:coded_elastic_computing}
\begin{algorithmic}
	\STATE {\bfseries Input:} The data matrix $\Xbf$, the number of machines $P$, the recovery threshold $L$, the linear combination coefficients $g_{s,k}$'s in equation \eqref{eqn:linear_combination} and the sequence of input vectors $\w_t$ $t=1,2,\ldots$.
    \STATE {\bfseries Preprocessing:} Partition the data $\Xbf$ into $L$ subsets and compute the coded subsets as in \eqref{eqn:linear_combination}.
    \STATE {\bfseries Online computation:}
    \STATE {\bfseries FOR} each computation with input $\w_t$:
    \bindent
   \STATE {\bfseries Broadcast:} The master node sends $\w_t$ to each worker.
   \STATE {\bfseries FOR} each group index $j$:

    \STATE {\bfseries \qquad Gather:} The $k$-th worker computes $\u_{t, k,j} = \Xbf_{k,j}^{\text{coded}}\w_{t}$ and sends $\u_{t, k, j}$ to the master.
    \STATE{\qquad The master gathers vectors $\u_{t, k, j}$ for all workers that use the $j$-th sub-block and obtains}
    \STATE{\qquad the matrix $\u_{t,j}$ which contains the results for the $j$-th group (row block).}
   \STATE {\bfseries \qquad Decode:} The master node computes $\u_{t,j}\Gbf_j^{-1}$ to obtain the results for the $j$-th group.
	
    \STATE {\bfseries Output:} The master node outputs $\Xbf\w_t$.

    \STATE {\bfseries IF Preemption/New Machines:} Change the selected data to use based on the current number of machines.
\eindent

\end{algorithmic}
\end{algorithm}

\subsection{Analysis of Coded Elastic Computing: Achieving Optimal Fully Transition Compatibility}

According to Definition~\ref{def:transition_family_optimal}, a fully compatible family of policies can support seamless transitions between any pairs of policies in the family. The coded elastic computing scheme provided in Algorithm \ref{alg:coded_elastic_computing} for matrix-vector multiplication gives a family of fully-compatible policies for fixed memory cost at each machine. In Section~\ref{sec:achievable}, we analyze the memory cost and the number of data points to be used at each existing machine. Then, in Section~\ref{sec:lower_bound}, we provide lower bounds on these two quantities and show that Algorithm~\ref{alg:coded_elastic_computing} achieves the optimal memory cost and the size of the selected data. Thus, the coded computing policies in Algorithm~\ref{alg:coded_elastic_computing} is a compatible family of \emph{optimal} policies, in that for the fixed storage size, each policy obtains the largest number of tolerable failures and smallest size of the selected data to use ($e^*$,$u^*$), and provides seamless transitions between each other without moving data at existing machines.

\subsubsection{Upper Bounds on the Storage Cost and the Size of the Slected Data to Use}\label{sec:achievable}
Suppose the original data has $N$ data points and all data points are in $\mathbb{R}^d$, i.e., $\Xbf$ has size $N\times d$. Recall that $P_\text{max}$ is the maximum number of machines, $L$ is the recovery threshold, and $n$ is the number of currently available machines. In the following theorem, by the \emph{size} of the data we mean the number of data points.

\begin{theorem}\label{thm:achievable}
Suppose the number of machines $n$ satisfies $L\le n\le P_\text{max}$ at any time, i.e., it is not smaller than the recovery threshold $L$ and not bigger than the maximum possible number of machines. Then, the coded elastic computing algorithm (Algorithm \ref{alg:coded_elastic_computing}) achieves the exact computation result of $\Xbf\w_t$ for all $t$. The size of the data stored at each machine is $N/L$. The size of the selected data to use at each machine is $N/n$ and is the same across different machines. The overall size of the selected data is the same as the size of the original data.
\end{theorem}
\begin{proof}
Recall that we call the sub-blocks on the same row block (in Figure~\ref{fig:elastic_computing}) of the $P$ different blocks of data a group, and we use $\Xbf_{k,j}^{\text{coded}}$ to represent the coded sub-block of data that is at the $k$-th machine and belongs to the $j$-th group. Let $\Xbf^j$ be the collections of original data that belongs to the $j$-th group. For example, in Figure~\ref{fig:sub2}, $\Xbf^6$ represents the collection of original data $\left[\begin{matrix}
\f_1, \f_2, \f_3
\end{matrix}\right]^\top$.

Then, from the cyclic way of using data, we can see that when each machine uses $L$ sub-blocks, the overall number of used sub-blocks in each group is $L$. From Lemma~\ref{prop:LoutofP}, the results $\u_{t, k,j} = \Xbf_{k,j}^{\text{coded}}\w_{t}$ for all the workers that use the coded data in the $j$-th group can be collected together to decode $\Xbf^j\w_{t}$.

The other claim can be seen from the figure, i.e., the area of the used data is always equal to the area of the original data, and the area of the used data is the same across all remaining machines. We prove this claim as follows. The size of the data in each of the $L$ subsets is $N/L$. Thus, the used data at each machine is the same number $\frac{N}{L}\cdot\frac{L}{n} = N/n$. There are $n$ machines left, so the overall size of the used data is $\frac{N}{n}\cdot n=N$, which is the same as the original data.
\end{proof}

Theorem \ref{thm:achievable} shows that our technique uses the same size of data as the original (uncoded) case. This is desirable for memory-bound applications.

\vspace{0.3ex}
\begin{remark} (Cost analysis)
Recall that $P$ is the number of workers, $N\times d$ is the size of the matrix $\Xbf$, $L$ is the recovery threshold, and $n$ is the number of currently available machines. The encoding (preprocessing step) is a one-time cost for online matrix-vector multiplications. The decoding by solving a linear system at the master node has computational cost $\Obf(LN)$, because the linear system for each group (row-block) of data involves $L$ equations on $L$ unknown subvectors of size $N/L/n$ (which is the height of each sub-block), and there are $n$ such groups. Thus, the computational cost using straightforward matrix-vector multiplication is $L^2\cdot (N/L/n)\cdot n = NL$. The matrix-multiplication step at each worker has cost $\Obf(dN/n)$. Thus, the decoding cost is smaller than the computational cost at each worker as long as $d=\Omega(nL)$. Even if $d<\Omega(nL)$, we can partition the machines into smaller groups and respectively code each group. The decoding complexity can be further reduced to $\Obf(N\log^2 L)$ if Vandermonde systems are used \cite{li2000arithmetic}, at the cost of numerical stability. One thing to note is that the decoding complexity, even for the straightforward matrix-vector multiplication method, is $NL$ and is independent of the number of workers $P$.
\end{remark}

\subsubsection{Lower Bounds on the Storage Cost and the Size of the Selected Data to Use}\label{sec:lower_bound}

Here, we provide a fundamental limit which shows that the achievable scheme provided in \ref{thm:achievable} is optimal in terms of the storage cost at each machine and the size of the actually used data at each iteration, for a fixed number of machines and a fixed number of tolerable machine preemptions. Before we present the theorem, we formalize the definition of the \emph{size} of data using number of bits. This is because in theory, we cannot store arbitrarily high-precision numbers.

\begin{assumption}\label{ass:finite_precision}
Suppose the entries of the matrix $\Xbf$ are i.i.d. random variables that take values in a finite set $\S\in\mathbb{R}$. Each of this random variable has entropy $H=\log |\S|$. Thus, the overall entropy of all the data in $\Xbf$ is $NdH$.
\end{assumption}

\begin{assumption}\label{ass:finite_precision_storage}
For a certain computation policy, suppose each machine initially stores an array of finite-precision numbers in its memory. Each finite-precision number can be an \emph{arbitrary} function of the original data. The overall number of finite-precision numbers that is stored is finite.
\end{assumption}

Note that although we use real-number computation all the time, the numbers that we deal with are always discrete, i.e., we conduct computation of finite-precision numbers in the real field. This validates Assumption \ref{ass:finite_precision}. Also note that the number of possible combinations of floating point numbers that can be stored by a finite-length bit array is finite. Therefore, Assumption \ref{ass:finite_precision_storage} is also valid. We need these two assumptions because if they do not hold, the system can concatenate all the real numbers in $\Xbf$ into one real number and only stores that particular real number. In that case, no bound on the storage is meaningful because one only needs to read this single number in memory to access all the information of $\Xbf$. We also need the assumption that the size of the stored numbers are finite because otherwise, we can enumerate all possible $\Xbf\w$ and store them.

\begin{assumption}\label{ass:data_not_move}
Suppose we do not alter the way that we store the data $\Xbf$ inside the memory after the computation begins. Even when a preemption-type failure happens, we do not move the data at existing machines.
\end{assumption}

\begin{assumption}\label{ass:using_data}
By \emph{selecting to use} one number stored in the memory, we mean the algorithm reads the whole number (e.g., reading all digits if the number is stored as a floating-point number) from the memory for further processing. The array of finite-precision numbers can only be accessed one number at a time, meaning that one cannot access a function value of two numbers and claim that only one number is selected to use.
\end{assumption}

\begin{theorem}\label{thm:fundamental}
Suppose the Assumptions \ref{ass:finite_precision}-\ref{ass:using_data} hold. Suppose we require the recovery of the exact computation result $\Xbf\w$. Then, the following fundamental limits hold.
\begin{itemize}
    \item [(a)] Denote the entropy of the encoded data at the $k$-th machine by $H_k,k=1,2,\ldots,n$. Then, to provide the tolerance to a maximum of $n-L$ failures (recall that $n$ is the current number of available machines), we have $\max_{k\in \{1,2,\ldots,n\}} H_k\ge \frac{N}{L}\cdot dH$;
    \item [(b)] The worst-case entropy of the actually used data (maximized with respect to the choice of $\w$) has to be no less than $NdH$, or $\frac{N}{n}\cdot dH$ at each machine.
\end{itemize}
\end{theorem}
\begin{remark}
By comparing Theorem~\ref{thm:achievable} and Theroem~\ref{thm:fundamental}, we see that the coded elastic computing technique in Algorithm~\ref{alg:coded_elastic_computing} achieves the fundamental limit because each data point has dimension $d$ and each entry has entropy $H$. We note a nuance here that the linear combinations in Algorithm~\ref{alg:coded_elastic_computing} may make each encoded number have entropy larger than $H$. However, this can be solved if (i) we approximate the real-number computation using the computations in a finite field and assume each number in the data matrix $\Xbf$ has a uniform distribution over the finite field \cite{yu2017polynomial}, or (ii) we assume that each entry in the data matrix comes from a quantized version of the standard Gaussian-distribution, and assume that the encoding matrix has normalized rows, such that the linear combination result in \eqref{eqn:linear_combination} (before quantization) is also from the standard Gaussian distribution. To make the point (ii) formal, suppose an arbitrary entry in the data matrix $\Xbf_k$ is sampled from i.i.d. standard Gaussian and quantized to $q(y_k)$ (here $y_k\sim \mathcal{N}(0,1)$). A linear combination of the samples has the form $\sum_{k=1}^L g_{s,k} q(y_k)$ (see equation \eqref{eqn:linear_combination}), where the linear combination coefficients satisfies $\sum_{k=1}^L g_{s,k}^2 = 1$. Then, the linear combination is stored as $q(\sum_{k=1}^L g_{s,k} q(y_k))$ in the $s$-th machine. Since the \emph{unquantized} sum $\sum_{k=1}^L g_{s,k} y_k$ is a standard Gaussian random variable, its quantized version $q(\sum_{k=1}^L g_{s,k} y_k)$ has the same entropy as each $q(y_k)$. Therefore, we only need to show that $q(\sum_{k=1}^L g_{s,k} q(y_k))$ has an entropy that is close to $q(\sum_{k=1}^L g_{s,k} y_k)$. This can be proved by bounding the difference in the entropy, because the two random variables $q(\sum_{k=1}^L g_{s,k} q(y_k))$ and $q(\sum_{k=1}^L g_{s,k} y_k)$ have the same support, and their pmf's can be made arbitrarily close when the quantization function $q$ has a quantization level $\Delta\to 0$.
\end{remark}
\begin{remark}
Note that the claim (b) has to be stated in a worst-case way because for many choices of $\w$, computing $\Xbf\w$ can be degenerated. For example, if we know in advance that $\w$ only takes value in a very small finite set of vectors, we can compute $\Xbf\w$ for all possible $\w$ and store these vectors. When $\w$ is sparse, we also do not need to read the entire matrix $\Xbf$. Therefore, we indeed need to state the fundamental lower bound in terms of the worst-case $\w$.
\end{remark}
\begin{proof}
Now, we prove statement (a). Suppose an arbitrary set of $n-L$ machines fail. Since the algorithm is tolerant to any $n-L$ failures, the master node can still recover exactly all the results $\Xbf\w$, no matter what $\w$ is. Therefore, we can choose $\w$ to be the elements of natural basis $\w = \e_i, i=1,2,\ldots, d$, and collect all the results $[\Xbf\e_1,\Xbf\e_2,\ldots,\Xbf\e_d] = \Xbf\I_d = \Xbf$. Since data processing can only reduce entropy (from the data processing inequality), the entropy of the overall data stored at the remaining $L$ machines is no less than the entropy of $\Xbf$ which is $NdH$. This holds for any combination of $L$ machines, i.e.,
\begin{equation}
    H(\Xbf_{i_1}, \Xbf_{i_2}, \ldots, \Xbf_{i_L})\ge NdH, \forall 1\le i_1<i_2<\ldots <i_L\le n.
\end{equation}
Adding up the above equation for all combinations of $L$ out of $n$ machines, and by plugging in $\sum_{j=1}^LH(\Xbf_{i_j})\ge H(\Xbf_{i_1}, \Xbf_{i_2}, \ldots, \Xbf_{i_L})$, we have $\sum_{k=1}^n H_k\ge \frac{NndH}{L}$.

Then, we prove statement (b). Suppose (b) is not true. Then, it means that there exists a way to encode and store the encoded data in the memory of the machines, such that for any arbitrary vector $\w$, computing $\Xbf\w$ only requires reading data of entropy strictly less than $NdH$. Since the overall number of stored data is finite, we can assume that the overall number of stored numbers is $p$. For an arbitrary subset $\S$ of the stored numbers such that $H(\S)<NdH$, we denote by $\Pd_\S$ the set of $\w$ such that $\Xbf\w$ is able to be computed using only the numbers in $\S$. Then, we can see that the $\Pd_\S$ for an arbitrary subset $\S$ is a linear subspace of $\mathbb{R}^d$. This is because if $\Xbf\w_1$ and $\Xbf\w_2$ are both able to be computed using the numbers in $\S$, then, $a\Xbf\w_1+b\Xbf\w_2 = \Xbf(a\w_1+b\w_2)$ is also able to be computed. Now, we prove that $\Pd_\S$ cannot be the entire $\mathbb{R}^d$ that $\w$ can take value from. This is because if $\Pd_\S$ is equal to $\mathbb{R}^d$, then $\Xbf\e_1,\Xbf\e_2,\ldots, \Xbf\e_d$ are all able to be computed using the numbers in $\S$ ($\e_1,\ldots \e_d$ are the standard basis), which means $\Xbf$ itself is able to be computed using $\S$. This is a clear contradiction to the data-processing inequality because $H(\Xbf)=NdH$, while $H(\S)<NdH$. Thus, $\Pd_\S$ can at most be a linear space of dimension $d-1$ in $\mathbb{R}^d$. This means the union of $\Pd_\S$ for all $\S$ is a finite collection of linear spaces of dimension $d-1$ in $\mathbb{R}^d$, which cannot cover the entire $\mathbb{R}^d$. Thus, there must exist $\w$ in $\mathbb{R}^d$ such that $\Xbf\w$ is not able to be computed by using data of entropy strictly less than $NdH$.
\end{proof}

\subsection{Related Works and the Comparison with Non-elastic Coded Computing Techniques}\label{sec:comparison}

Coded computing is an emerging area to tackle stragglers \cite{suh2017matrix,yang2017NIPS,ferdinand2016anytime,ferdinand2018hierarchical,park2018hierarchical,tang2018erasure,sheth2018application,ozfaturay2018speeding,maity2018robust,baharav2018straggler,lee2017multicore,yu2018lagrange,Emina2,heterogeneousclusters,lee2017matrix,lee2018speeding,dutta2018optimal,wang2018coded,mallick2018rateless,Suhas1,Suhas2,severinson2017block,yu2018entangled,dutta2016short,dutta2017coded,fahim2017optimal}, machine failures and soft errors \cite{dutta2018DNN1,Virtualization,yang2017encoded,hadjicostis2005coding,jeongLRC}, security issues and adversaries \cite{bitar2017minimizing,bitar2018staircase,charles2018gradient,yang2019secure,yu2018lagrange,chen2018draco}, and communication bottlenecks \cite{GC2,GC3,GC4,ye2018communication,jeongFFT,Salman1,allerton18,jeongFFT,prakash2018coded}. It can handle both linear and nonlinear computations \cite{kosaian2018learning,dutta2018DNN1,GC2}. It is a significant advance on classical algorithm-based fault-tolerance (ABFT) techniques \cite{ABFT1984} and noisy computing \cite{Tay_Bel_68,Pip_FOC_85,von1956probabilistic}, and often achieves scaling sense improvements. Our work is the first to address the system elasticity issue using coded computing.

Non-elastic coded computing cannot adapt the algorithmic procedures to exploit \textit{new} machines. We consider the case when the number of machines increases from $L$ (i.e., the recovery threshold) to infinity with fixed storage size at each machine. We can easily see that elastic coded computing can have lower and lower memory-access time by using less and less data at each machine, while ordinary coded computing cannot because the used data remains unchanged at each machine. However, each node can have more than one linear combinations \cite{mallick2018rateless}. E.g., suppose $L$=8 linear combinations suffice to recover the result, and each machine has 8 linear combinations. Then, when the number of machines increases from 1 to 8, rateless coded computing can move between the optimal points $(n,m)=$ $(1,8)$, $(2,8)$, $(4,8)$, $(8,8)$, i.e., it achieves the optimal size of the selected data to use at some of the $(n,m)$ configurations. But they cannot move beyond the point $(8,8)$ and use more machines, because each machine, in the rateless-coded case, cannot use smaller than 1 (a fraction) of linear combinations \footnote{Although for rateless code, the number of available machines cannot be more than the code dimension $L$. Note that in this example, the number of available machines in the coded elastic computing scheme cannot exceed the memory size $m$ either. However, since $m$ is usually in scaling sense larger than the number of machines, this limit will not be reached unless in the limit of extreme strong scaling.}. From another perspective, for the rateless-coded scheme to achieve fully transition compatibility in the large scale, the number of linear combinations at each machine also has to grow with the number of (possible) machines, leading to scaling-sense higher decoding complexity than the elastic partitioning technique in Algorithm~\ref{alg:coded_elastic_computing} (Note that the decoding cost in Algorithm~\ref{alg:coded_elastic_computing} does not increase with the number of machines, and can be further reduced if parallel decoding is allowed. See Section~\ref{sec:fully-distributed} for details). Another advantage of coded elastic computing compared to the rateless-coded scheme is that it allows non-uniform partitioning of each row-strip (see Figure~\ref{fig:elastic_computing}) and can have much more fine-grained task scheduling. Regarding the communication efficiency of coded elastic computing, see Section \ref{sec:fully-distributed} for a fully-distributed version of coded elastic computing.

For task-level failures or stragglers, dynamic task allocation is useful \cite{gauri2015straggler, gauri2014efficient}. The task scheduler can choose which tasks to replicate, relaunch based on task profiles, and delay the relaunching to save time \cite{Emina1,ananthanarayanan2013effective,Emina2,gardner2015reducing,suresh2015c3}. However, unlike task relaunching and allocation, recovering from machine preemptions requires both restoring the machine states and downloading the data, which is time-consuming. The situation gets worse if elastic events are frequent.

\section{Applications of the Coded Elastic Computing}\label{sec:more_applications}

The cyclic way of elastic data partitioning applies to general coded computing techniques proposed thus far and is not limited to matrix-vector multiplications.

\subsection{Matrix-Matrix Multiplications}\label{sec:mat-mat}

First, we consider the application of matrix-matrix multiplications. We consider an \emph{online} version: we store an encoded version of the matrix $\Abf$ at $P$ machines and compute the matrix multiplication $\Abf\Bbf$ for different $\Bbf$'s. It is shown \cite{fahim2017optimal,dutta2018optimal,yu2018entangled} that the storage-optimal technique for coded robust matrix multiplications is the MatDot scheme proposed in \cite{fahim2017optimal,dutta2018optimal}. This technique partitions the matrix $\Abf$ column-wise and the matrix $\Bbf$ row-wise, and stores linearly combined submatrices of $\Abf$ and $\Bbf$ at each machine. In an online setting with elastic machine preemptions, we encode and store $\Abf$ at the initial stage, and do not move $\Abf$ anymore. When we receive $\Bbf$, we still partition the matrix $\Bbf$ row-wise but linearly combine them using the knowledge of the availability of the machines. The advantage is that we do not need to use the polynomial-based codes and can thus avoid possible numerical issues. At the same time, since the availability of the machines are known before we encode $\Bbf$, we can also remove the factor of 2 in the recovery threshold of MatDot codes. We can also use the same computational time cost as the uncoded case, which is similar to what we can achieve in the matrix-vector case.

More precisely, suppose we parition the matrix $\Abf$ into $L$ column blocks $\Abf_1, \ldots, \Abf_L$, and encode these blocks initially into $P$ blocks $\Abf_s^\text{coded}, s=1,2,\ldots,P$, where $P$ is the initial number of machines:
\begin{equation}
    \Abf_s^\text{coded} = \sum_{k=1}^L g_{s,k} \Abf_k.
\end{equation}
Denote by $\Gbf_{P_\text{max}\times L}$ the predetermined encoding matrix $[g_{s,k}]$, which is of size $P_\text{max}\times L$. Again, the extra linear combinations $\Abf_s^\text{coded}, s=P+1, P+2, \ldots,P_\text{max}$ can be generated offline before the computation and stored on the cloud. Assume there are currently $n$ machines that remain. Similar to Algorithm~\ref{alg:coded_elastic_computing}, we partition each coded sub-matrix $\Abf_s^\text{coded}$ into $n$ submatrices. However, here, we partition each $\Abf_s^\text{coded}$ \emph{column-wise} to $\Abf_{s,i}^\text{coded}, i=1,2,\ldots, n$. This is mathematically equivalent to partitioning each uncoded submatrix $\Abf_s$ into $\Abf_{s,i}, i=1,2,\ldots, n$, and compute
\begin{equation}
    \Abf_{s,i}^\text{coded} = \sum_{k=1}^L g_{s,k} \Abf_{k,i}, i=1,2,\ldots, n.
\end{equation}
The subscript $s$ in $\Abf_{s,i}$ belongs to a subset with size $n$ of the set $\{1,2,\ldots,P_\text{max}\}$, which corresponds to the $n$ available machines after the preemption failures. We again select to use these submatrices $\Abf_{s,i}$ in a cyclic fashion, just as shown in Figure~\ref{fig:elastic_computing}. For example, consider the case when $P=6$, $L=3$, $n=4$ and the 2nd and the 4th machines are preempted, which is exactly the same as shown in Figure~\ref{fig:sub4}. Then, we use $\Abf_{1,1}, \Abf_{1,2}, \Abf_{1,3}$ at the 1st machine, $\Abf_{3,2}, \Abf_{3,3}, \Abf_{3,4}$ at the 3rd machine, and so on. Denote by $\mathcal{S}_i$ the set of the indices of all the machines that use $\Abf_{s,i}$. For example, for $i=1$, in the above example, $\mathcal{S}_1 = \{1,5,6\}$. Again, similar to Section~\ref{sec:algorithm}, denote by $\Gbf_i$ the submatrix of $\Gbf$ with row indices in $\mathcal{S}_i$. Each $\Gbf_i$ is a $L\times L$ matrix. Denote by
\begin{equation}\label{eqn:H_i}
 \Hbf_i = (\Gbf_i^\top)^{-1}  ,
\end{equation}
and assume $\Hbf_i$ has element $h_{l,k},l = 1,2,\ldots,L,k=1,2,\ldots,L$.

When we get the matrix $\Bbf$, we partition it row-wise into $L$ blocks $\Bbf_1, \Bbf_2, \ldots, \Bbf_L$. Then, we partition each submatrix $\Bbf_k$ into $n$ row blocks $\Bbf_{k,i}, i=1,2,\ldots,n$ as well. For each $i$, suppose $\mathcal{S}_i = \{s_{i,1},s_{i,2},\ldots, s_{i,L}\}$, where $1\le s_{i,1}<s_{i,2}<\ldots< s_{i,L}\le n$. Then, we encode
\begin{equation}\label{eqn:encode_B}
    \Bbf_{s_{i,l},i}^\text{coded} = \sum_{k=1}^L h_{l,k}\Bbf_{k,i},
\end{equation}
and send $\Bbf_{s_{i,l},i}^\text{coded}$ to the $s_{i,l}$-th machine to compute $\Abf_{s_{i,l},i}^\text{coded}\Bbf_{s_{i,l},i}^\text{coded}$.

Now, we show that when we do a reduction on all the partial results $\Abf_{s_{i,l},i}^\text{coded}\Bbf_{s_{i,l},i}^\text{coded}$, we can indeed get $\Abf\Bbf$.
\begin{lemma}
If we use the cyclic partitioning technique to determine each $\mathcal{S}_i = \{s_{i,1},s_{i,2},\ldots, s_{i,L}\}$ and encode $\Bbf$ according to \eqref{eqn:encode_B}, we have
\begin{equation}
\sum_{i=1}^{n} \sum_{l=1}^L \Abf_{s_{i,l},i}^\text{coded}\Bbf_{s_{i,l},i}^\text{coded} = \Abf\Bbf.
\end{equation}
\end{lemma}
\begin{proof}
Note that
\begin{equation}
\begin{split}
\Abf_{s_{i,l},i}^\text{coded}\Bbf_{s_{i,l},i}^\text{coded} =&   (\sum_{k=1}^L g_{s_{i,l},k} \Abf_{k,i})(\sum_{k=1}^L h_{l,k}\Bbf_{k,i}) \\
\overset{(a)}{=}& [\Abf_{1,i},\Abf_{2,i},\ldots,\Abf_{L,i}] (\Gbf_i)_\text{l-th row}^\top (\Hbf_i)_\text{l-th row} [\Bbf_{1,i}^\top,\Bbf_{2,i}^\top,\ldots,\Bbf_{L,i}^\top]^\top,
\end{split}
\end{equation}
where (a) holds because $[g_{s_{i,l},1}, g_{s_{i,l},2}, \ldots, g_{s_{i,l},L}]$ is the $s_{i,l}$-th row in $\Gbf$, which is the $l$-th row in $\Gbf_i$ (recall that $\Gbf_i$ is the submatrix of $\Gbf$ with row indices in $\mathcal{S}_i = \{s_{i,1},s_{i,2},\ldots, s_{i,L}\}$). Thus, by adding up the above equation for all $l$, we have that for each $i$,
\begin{equation}
\begin{split}
    \sum_{l=1}^L \Abf_{s_{i,l},i}^\text{coded}\Bbf_{s_{i,l},i}^\text{coded}
     =&\sum_{l=1}^L[\Abf_{1,i},\Abf_{2,i},\ldots,\Abf_{L,i}] (\Gbf_i)_\text{l-th row}^\top (\Hbf_i)_\text{l-th row} [\Bbf_{1,i}^\top,\Bbf_{2,i}^\top,\ldots,\Bbf_{L,i}^\top]^\top\\
    =& [\Abf_{1,i},\Abf_{2,i},\ldots,\Abf_{L,i}] \Gbf_i^\top \Hbf_i [\Bbf_{1,i}^\top,\Bbf_{2,i}^\top,\ldots,\Bbf_{L,i}^\top]^\top\\
    \overset{(a)}{=}& [\Abf_{1,i},\Abf_{2,i},\ldots,\Abf_{L,i}] [\Bbf_{1,i}^\top,\Bbf_{2,i}^\top,\ldots,\Bbf_{L,i}^\top]^\top \\
    =& \sum_{k=1}^L \Abf_{k,i}\Bbf_{k,i},
\end{split}
\end{equation}
where the equation (a) holds because $\Hbf_i = (\Gbf_i^\top)^{-1}$ (see equation \eqref{eqn:H_i}). Adding up the above equation for all $i$, we have
\begin{equation}
    \sum_{i=1}^{n} \sum_{l=1}^L \Abf_{s_{i,l},i}^\text{coded}\Bbf_{s_{i,l},i}^\text{coded} = \sum_{i=1}^{n}\sum_{k=1}^L \Abf_{k,i}\Bbf_{k,i} = \Abf\Bbf.
\end{equation}
\end{proof}
\begin{remark}(Cost analysis)
We can see that the computational cost for this algorithm comes from two parts: (1) the encoding of $\Bbf$ using \eqref{eqn:encode_B}, and (2) the computational cost of $\Abf_{s_{i,l},i}^\text{coded}\Bbf_{s_{i,l},i}^\text{coded}$. Assume the matrix $\Abf$ has size $d_A\times N$ and $\Bbf$ has size $N\times d_B$, where $d_A,d_B = \Theta(N)$. Then, encoding $\Bbf$ has complexity $nL\cdot L\cdot (d_B\times N/L/n)= Ld_BN = \Theta(LN^2)$. Computation of a single $\Abf_{s_{i,l},i}^\text{coded}\Bbf_{s_{i,l},i}^\text{coded}$ has complexity $d_A\times N/L/n \times d_B$, and each machine computes $L$ different $\Abf_{s_{i,l},i}^\text{coded}\Bbf_{s_{i,l},i}^\text{coded}$, which means the overall complexity is $d_A\times N/L/n \times d_B\times L = d_Ad_BN/n$. Note that this is in the order of $\Theta(N^3)$ and the complexity is the same as distributing the matrix-matrix multiplication task to $n$ machines. The encoding time is much less than the computation time per worker if $Ld_BN\ll d_Ad_BN/n$, or $Ln\ll d_A$. If the encoding time is much smaller than the matrix-multiplication time, we again achieve linear scaling of the theoretical computational complexity in the number of available machines $n$ without moving data at the existing machines. At the same time, the computation time cost is the same at each worker machine as the uncoded case, which means that the factor 2 in the MatDot codes \cite{dutta2018optimal} can be removed\footnote{It was shown that the optimal recovery threshold in MatDot is $2L-1$, which means that the computational cost at each machine is twice that of the uncoded case for the same number of machines. The scheme here removes the factor of 2 because we consider a different problem: the system becomes aware of the indices of the failed machines after the preemptions have happened, and can adaptively change the encoding of $\Bbf$.}. 
\end{remark}

\subsection{Coded Elastic Computing for Linear Models}

Then, we focus on the application of coded computing for linear regression \cite{lee2018speeding,li2018polynomially}. For the ease of presentation, we consider vanilla gradient descent (we also use line search in the experiment validation for all competing techniques), in which the full matrix $\Xbf$ is used at each iteration. The technique developed here naturally generalizes to stochastic gradient descent and other generalized linear models such as logistic regression. Consider the linear objective function:\vspace{-3mm}
\begin{equation}
   f(\w; \Xbf, \y) = \sum_{i=1}^n (\w^\top \x_i -\y_i)^2 + h(\w).
\end{equation}
The vanilla distributed gradient descent has the form $\w_{t+1} = \w_t -\eta \g_t$ and $\g_t = \Xbf^\top (\Xbf \w_t - \y) + \partial_{\w} h(\w_t)$. When the data matrix $\Xbf$ is large, the most time-consuming part is the computation of $\Xbf^\top (\Xbf \w_t - \y)$. We thus extend Algorithm~\ref{alg:coded_elastic_computing} in the following way to compute $\Xbf^\top (\Xbf \w_t - \y)$. It is nothing but a combination of Algorithm~\ref{alg:coded_elastic_computing} and the matrix-matrix multiplication algorithm in Section~\ref{sec:mat-mat}.
\begin{itemize}
    \item Compute $\Xbf^j \w_t$ (where $\Xbf^{j}$ is the data in the $j$-th group, or the $j$-th row block in Figure~\ref{fig:elastic_computing} across different machines) for all group-index $j$ in an elastic way using Algorithm~\ref{alg:coded_elastic_computing};
    \item The master computes $\z_t^j = \Xbf^j \w_t - \y^j$ for all group-index $j$, where $\y^j$ are the labels corresponding to the data points in the $j$-th group;
    \item The master re-encodes $\z_t^j$'s using the (pre-computed) inverse generator matrix $\Hbf_j = (\Gbf_j^{-1})^\top$ to obtain $(\Gbf_j^{-1})^\top\z_t^j$, and scatters the results to the workers that use the $j$-th group of data;
    \item Since data at workers are encoded using $\Gbf_j$, the reduced results from all the workers are
    \begin{equation}
       \sum_j (\Xbf^{j})^\top \Gbf_j^\top(\Gbf_j^{-1})^\top\z_t^j = \sum_j (\Xbf^{j})^\top (\Xbf^j \w_t - \y^j) = \Xbf^\top (\Xbf \w_t - \y).
    \end{equation}
\end{itemize}
\vspace{-1ex}
Note that in the above extension of Algorithm~\ref{alg:coded_elastic_computing}, the workers also utilize the data as in Figure~\ref{fig:elastic_computing}. The experiment results of coded elastic computing in linear models are provided in Section~\ref{sec:exp_MT}.

\subsection{Fully Distributed Coded Elastic Computing}\label{sec:fully-distributed}

Coded elastic computing is not restricted to a master-worker setting. In this section, we consider a fully distributed coded elastic matrix-vector multiplication technique that is a trivial generalization of Algorithm~\ref{alg:coded_elastic_computing}. The advantage of a fully distributed framework is that communication among workers can be overlapped and the communication to the master does not become the single bottleneck in the limit of a large number of machines. Consider an application of iterative matrix-vector multiplication
\begin{equation}
\w_{t+1} = f(\Xbf\w_t),
\end{equation}
where $\Xbf$ is a square data matrix, and $f(\cdot)$ is an entry-wise low-complexity operation on the vector $\Xbf\w_t$. We again consider the example shown in Figure~\ref{fig:elastic_computing},\footnote{Note that the example in Figure~\ref{fig:elastic_computing} partitions the data row-wise. The column-wise partition can use the elastic coding scheme in Section~\ref{sec:mat-mat}, which is essentially an elastic dot-product scheme.} but apart from the 6 worker machines, there is no master-node. In each iteration, each machine computes its own matrix-vector multiplication based on the selected data to use. For example, the 2nd machine computes $\b_2 \w_t$, $\c_2 \w_t$, $\d_2 \w_t$, and stacks the results into one vector. Then, the results at each machine can be broadcast to all the other machines using an all-gather communication (using a bucketing algorithm \cite{barnett1994interprocessor}). The communication is also done in a cyclic way so that all the communications can be maximally overlapped (i.e., all machines can communicate at the same time). In this way, the communication bandwidth can be reduced when compared to the master-worker framework. Moreover, the overall communication time in this scheme is the same as the uncoded scheme using the same number of machines, because the selected size of data at each machine is the same as the uncoded case. The decoding of the intermediate results at each iteration is done independently at each machine. The decoding can also be conducted distributedly for each group of data (each row-block), i.e., each machine only takes care of the decoding of each row-block. This can significantly reduce the decoding time, but requires all machines to exchange different messages twice and carry out an all-to-all communication \cite{bruck1997efficient}. Also, the smaller communication is possible at the cost of more storage or computation \cite{agarwal1995three,solomonik2011communication,Salman1,ye2018communication} in the individual machines .

\subsection{Application to Deep Neural Networks}\label{sec:dnn}

The proposed coded elastic data partitioning directly applies to the coded training of model-parallel deep neural networks \cite{dutta2018DNN1,dutta2018DNN2}. The coded computation in \cite{dutta2018DNN1,dutta2018DNN2} utilizes a novel technique called PolyDot, which adopts a 2D partitioning on the weight matrices of a neural network and encodes the weight matrices using polynomial-based codes. The way to apply coded elastic data partitioning to the PolyDot technique is exactly the same as Section \ref{sec:coded_elastic_partitioning}:
\begin{itemize}
    \item In PolyDot, each machine has one encoded submatrix of the weight matrix. Partition each encoded submatrix in Polydot into $n$ small submatrices. Recall that $n$ represents the number of currently available machines.
    \item Let each machine select to compute the matrix-vector multiplication using only $L$ out of $n$ small submatrices. Recall that $L$ is the recovery threshold.
    \item The selection is done in a cyclic way as well to ensure that there are enough encoded computation results to decode the original computation results.
\end{itemize}
Note that the matrix-vector multiplication in \cite{dutta2018DNN1,dutta2018DNN2} requires both forward pass and backward pass. This means that for the weight matrix $\mathbf{W}$, we need to compute matrix-vector multiplications in two ways, i.e., $\mathbf{Wx}$ and $\mathbf{W}^\top \mathbf{x}$. Therefore, in the forward pass, the partitioning of the encoded weight submatrix at each machine should be done in a row-wise way, which is similar to Section \ref{sec:coded_elastic_partitioning}. In the backward pass, the partitoining should be done in a column-wise way because one needs to compute $\mathbf{W}^\top \mathbf{x}$.

\eat{
\section{Mini-Benchmark on Amazon EC2}

In this section, we show an experiment using the Amazon Elastic Compute Cloud. We use one master machine and six worker machines. We use the t2.medium instances. We mimic the environment of the elastic computing by using different number of machines to compute the same matrix-vector product $\mathbf{Xw}$. The matrix has size $30000\times 10000$ and it is partitioned initially into 6 submatrices of size $5000\times 10000$. Then, each submatrix is stored at one machine. To mimic the elastic events, we change the number of available machines by injecting artificial failures. The maximum number of failures is 3. The per-iteration overall time (including both communication and computation) is shown in Figure~\ref{fig:amazon_minitest}. The result is averaged using 20 independent trials.
}

\section{Implementation and Experimental Evaluation}\label{sec:exp_MT}

The proposed coded elastic computing technique has been implemented on top of Apache REEF~\cite{reef} Elastic Group Communication (EGC) framework~\footnote{\url{https://github.com/interesaaat/reef/tree/elastic-sync}}. REEF EGC provides an API allowing to implement elastic computations by chaining fault-tolerant MPI-like primitives.
In this short paper we assess the performance of our elastic code computing approach through 2 mini-benchmarks.

\stitle{Matrix-vector mini-benchmark.}
In this mini-benchmark,  we test that indeed the time cost decreases linearly with the increase in the number of machines available.
We mimic an elastic computing environment on Amazon EC2 by using different numbers of M4.large instances to compute the same matrix-vector product $\mathbf{Xw}$. The master node is a M4.4xlarge machine. The matrix is randomly generated and with size $100000\times 10000$, and it is partitioned initially into $10$ submatrices of size $10000\times 10000$. Then, they are encoded into $20$ submatrices of the same size, and each submatrix is stored at one machine (for a total of 20 machines). To mimic the elastic events, we change the number of available machines by injecting artificial failures. The maximum number of failures is 10. The per-iteration overall time (including both communication and computation) is shown in Figure~\ref{fig:amazon_minitest}. The result is averaged using 20 independent trials.
As we can see, the coded elastic computing technique can utilize the extra machines when the number of machines increases (see the blue bars). We also compare it with a non-adaptive coded computation baseline. It can be seen that when the number of machines increases, the adaptive scheme can provide increasing speedup when compared to the non-adaptive coded computing scheme\footnote{Note that in the non-adaptive scheme, we do not utilize the coding to deal with stragglers because the amount of stragglers in m4 instances is small. The per-iteration time increases slightly with the number of machines because of the extra communication overhead}. The maximum observed speedup in our experiment is 46\%.

\stitle{Linear model mini-benchmark.} In this experiment, we test a coded implementation of linear regression using line-search-based batch gradient descent (the same setting as the baseline \cite{narayanamurthy2013towards}). We run the test over 20 machines on a Microsoft internal multi-tenancy cluster. Each data point in the dataset has 3352 features, and we sample 10000 data for training and 10000 data for testing. We generate random failures and allow REEF EGC to reschedule new machines when failures occur. 
\begin{figure*}
    \centering
    \subfigure[Matrix-vector time]{
		\includegraphics[width=60mm]{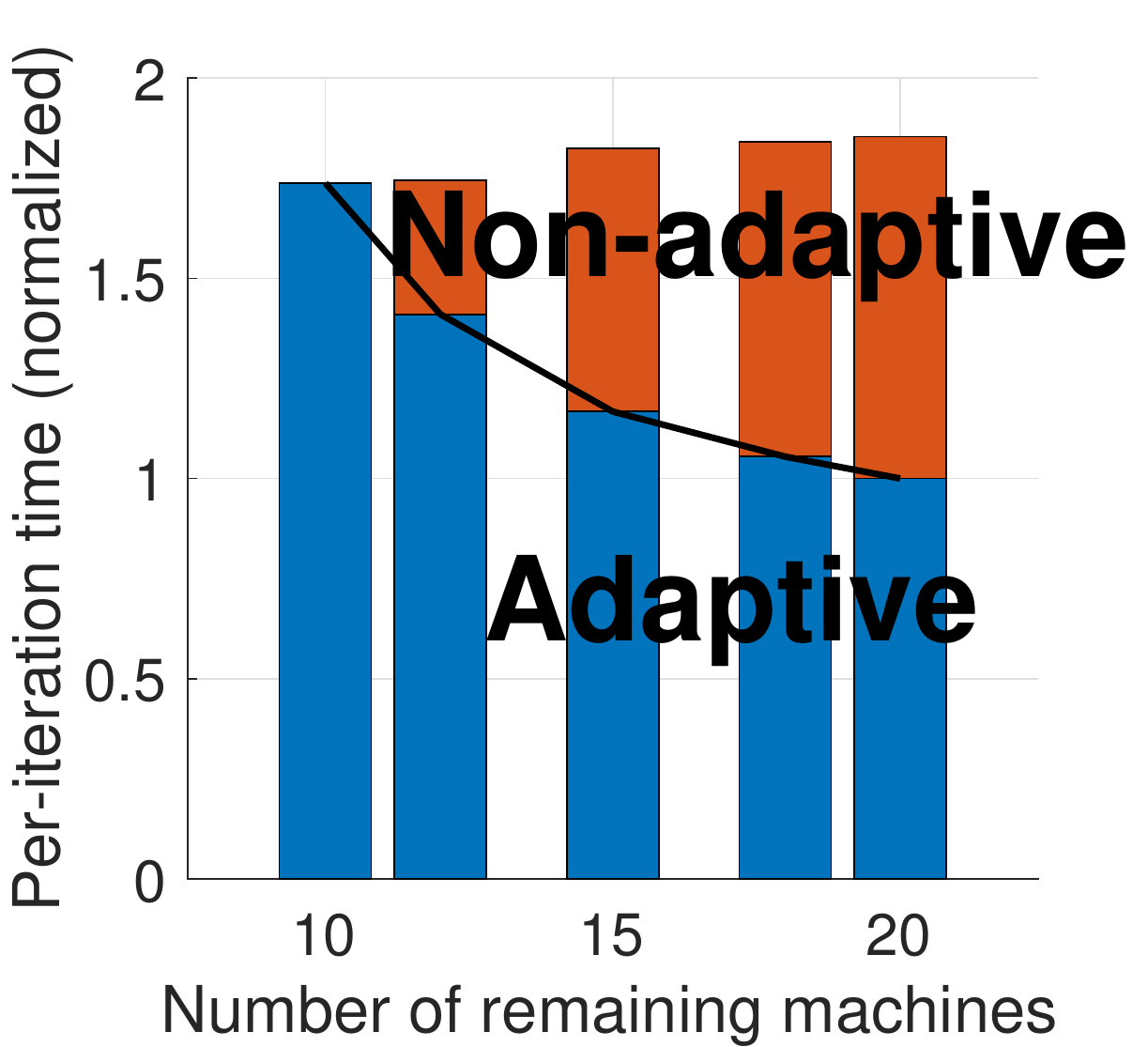}
		\label{fig:amazon_minitest}}
	\subfigure[Linear model time]{\includegraphics[width=53mm]{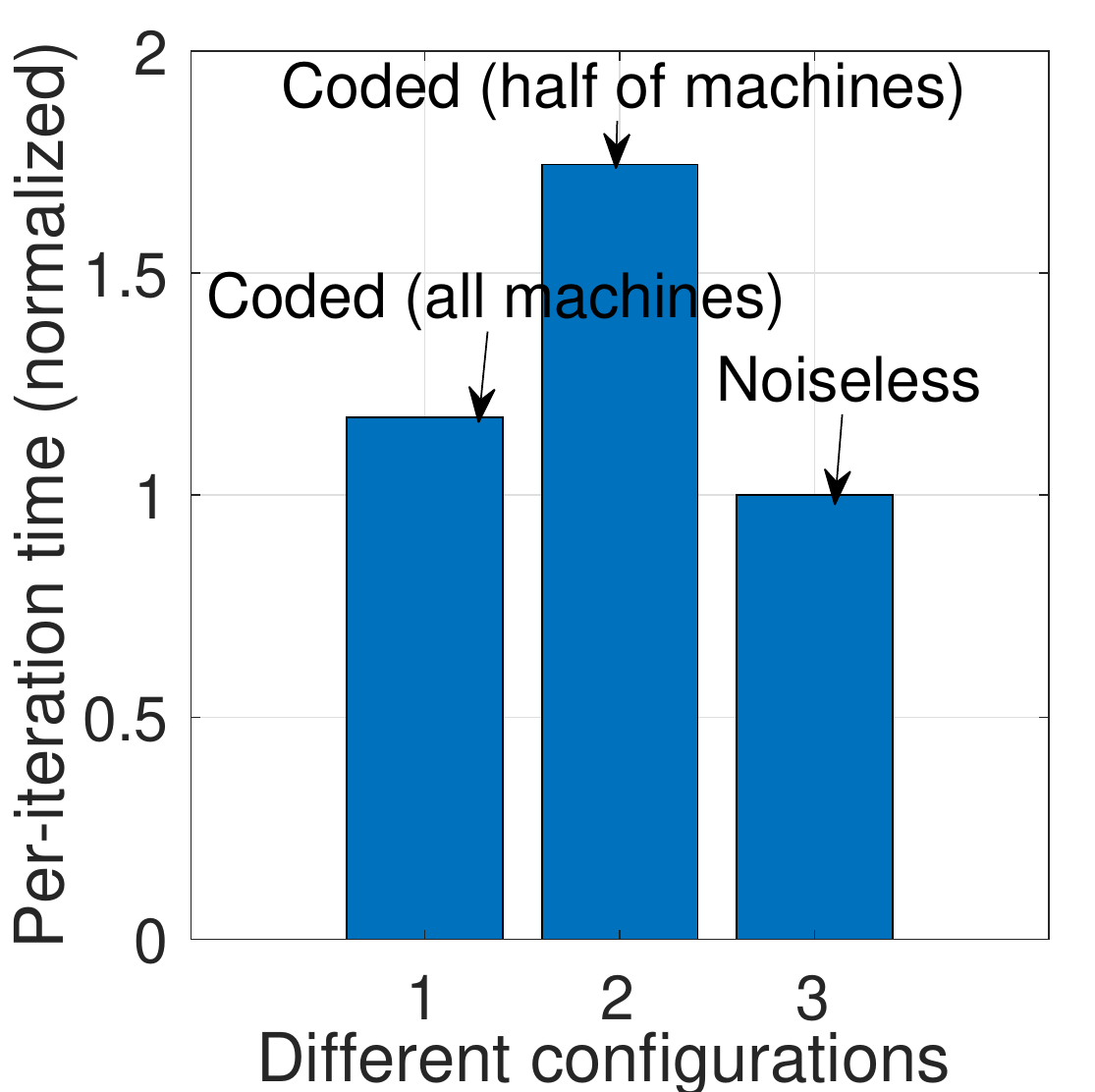}\label{fig:MT1}}
    \subfigure[Linear model error (different methods)]{\includegraphics[width=65mm]{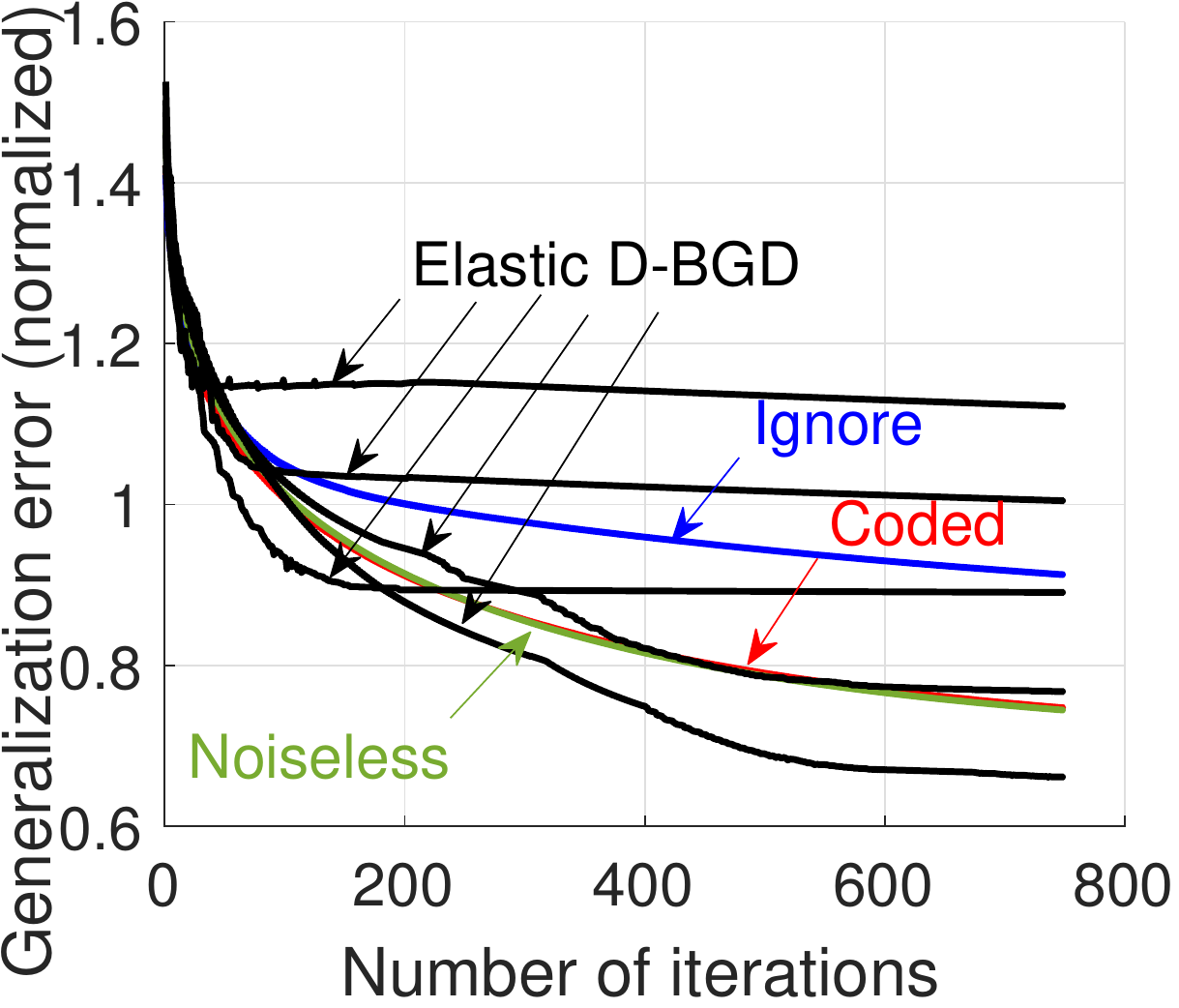}\label{fig:MT2}}
	\subfigure[Linear model error (different regularization coefficients)]{\includegraphics[width=65mm]{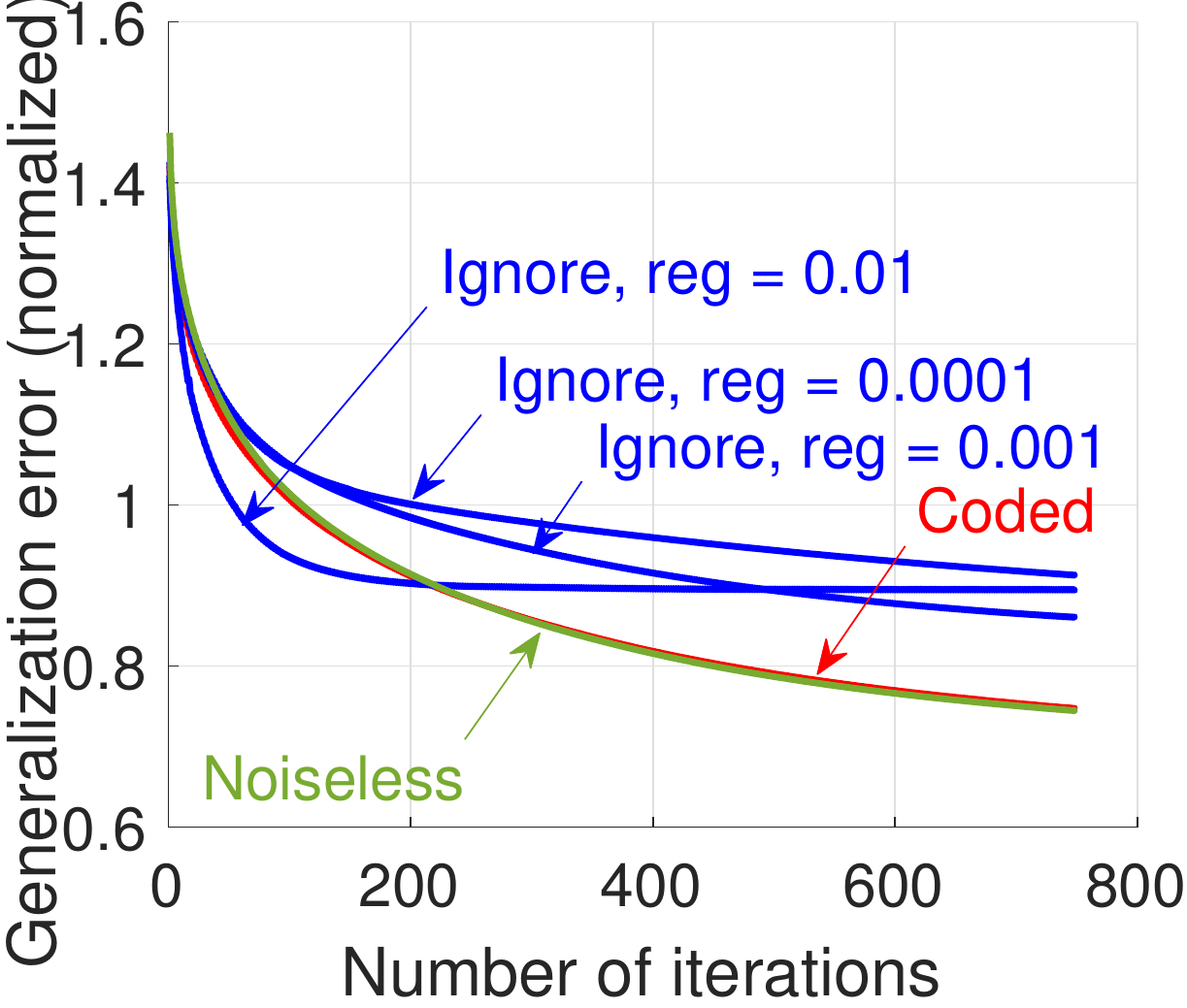}\label{fig:MT3}}
	\caption{Mini-benchmarks experiments (results normalized due to confidentiality). The per-iteration time in \ref{fig:amazon_minitest} and \ref{fig:MT1} includes decoding and communication.\vspace{-2mm}}
\end{figure*}
We start with Figure~\ref{fig:MT1} where we plot the time for each iteration. In theory, when all the workers are present, the computational cost per iteration should be the same as the uncoded case. However, the coded method (all) has slight overhead due to coding and a second round of communication for coded linear regression using matrix-vector multiplications. The coded method (half) shows the cost when only half of the workers are running, which is, as expected, twice the cost of the uncoded method.
In Figure~\ref{fig:MT2} and \ref{fig:MT3} we report the generalization error and we compare our coded elastic computing technique with three baselines, namely noiseless (no failure), ignore the failure and continue, and an existing algorithm called Elastic Distr-BGD \cite{narayanamurthy2013towards}. The coded method can achieve the same convergence behavior as the noiseless case, while the ignore method achieves worst generalization error even for different regularization parameters (the other techniques all use the same regularization\footnote{Note that although a fine-grained grid search is possible to find the best regularization coefficient, it is very time consuming in practice.} coefficient 0.001). In Figure~\ref{fig:MT2}, we show 5 different experiments on Distr-BGD using the same failure probability but different realizations. The convergence of Distr-BGD depends on when a failure occurs and can lead to different algorithm performance. This is because the Distr-BGD keeps using previous gradient vectors at the failed machines, and this can (1) lead to overfitting, and (2) make the optimization miss the minimum point. In the plot of Distr-BGD, the \emph{valley} part is due to overfitting, and the sudden change to near flat loss growth is because when the gradient descent has missed the optimal point of empirical training loss, the fixed gradient at the failure nodes makes the line search choose the smallest step size. In some cases, the Distr-BGD works extremely well because the fixed gradients act like momentum and can improve the speed of convergence.

From the experiment results, we can see that the coded elastic computing technique can obtain the same convergence behavior as ordinary gradient-descent-based algorithms but can elastically allocate the workload based on the number of available machines without moving data around.

\section{Conclusions}
\vspace{-1ex}
The coded elastic computing framework presented in this paper can deal with new cloud offerings where machines can leave and join during the computation.
Our framework handles the elastic events in a positive way, meaning that when machines leave, it shifts the computation to the remaining workers, and when new machines join the computation, it actively reduces the workload of existing machines without the reallocation of data. We prove that the coded elastic computing technique can achieve the same memory-access cost as the noiseless case, and hence is optimal for memory-bound applications. Using experiments in both Amazon EC2 and on a Microsoft multi-tenancy cluster, we show that the coded elastic computing technique can achieve the same convergence behavior as if no failure occurs, and can dynamically adjust working loads respect to the number of remaining workers. The proposed technique can be applied to coded matrix-vector, matrix-matrix multiplications and linear regression, and potentially other applications where the large-scale matrix operations are the bottleneck.
\eat{
The active way of utilizing remaining workers in the proposed technique will be useful in the case of elastic \emph{ramp-up} in which the learning task can start before all the machines are ready \cite{reef,hindman2011mesos}. The task continues when new machines are ready, and the workload at each machine reduces. Although we have not explicitly provided the algorithm that utilizes elastic ramp-up, the current framework already supports this function.
}

\section{Acknowledgment}

We sincerely thank Professor Michael Mahoney, Professor Kannan Ramchandran, Professor Viveck Cadambe, Professor Joseph E. Gonzalez, Vipul Gupta, Swanand Kadhe, Zhewei Yao, Sanghamitra Dutta, and Haewon Jeong for helpful discussions and suggestions. 

\bibliographystyle{abbrv}
\bibliography{bib}

\end{document}